\documentclass[12pt]{article}
\usepackage[dvips]{color} 
\usepackage{amsmath, amsthm, amssymb}
\usepackage[left=1.5in,top=1in,right=1in,nohead]{geometry}
\usepackage[title,toc,titletoc,page]{appendix}
\usepackage{latexsym}
\usepackage{mathrsfs}
\usepackage{rawfonts}
\usepackage[dvips,dvipdfmx]{graphicx}
\usepackage{float}
\usepackage[all]{xy}
\usepackage[sc,osf]{mathpazo}  

\usepackage{amsmath,amsthm,amssymb}
\usepackage{hyperref}
\usepackage{latexsym}
\usepackage{arydshln}
\usepackage{bmpsize}

\usepackage[utf8]{inputenc}
\usepackage{xcolor}
\usepackage{amsmath,amssymb}
\usepackage{appendix}
\usepackage{graphics}
\usepackage{float}
\usepackage{authblk}

\theoremstyle{plain}

\newtheorem{corollary}{Corollary}

\newtheorem{definition}{Definition}

\newtheorem{lemma}{Lemma}

\numberwithin{equation}{section}

\newcommand{\eone}{\mathbf{e_1}}
\newcommand{\etwo}{\mathbf{e_2}}
\newcommand{\ethree}{\mathbf{e_3}}

\DeclareMathOperator{\sech}{sech}

\begin{document}

\title{Classification of Killing Magnetic Curves In $\mathbb H^3$}

\author[1]{Özgür Kelekçi\footnote{Corresponding author}}
\author[2]{Furkan Semih Dündar}
\author[3]{Gülhan Ayar}
\affil[1]{{\small Department of Basic Sciences \& Faculty of Engineering,  \hspace{125mm} University of Turkish Aeronautical Association,  Ankara,  Turkey \hspace{125mm} okelekci@thk.edu.tr}}

\affil[2]{{\small Department of Mechanical Engineering, Amasya University \hspace{225mm}  Amasya 05100, Turkey \hspace{225mm}  furkan.dundar@amasya.edu.tr}}

\affil[3]{{\small Department of Mathematics, Karamanoğlu Mehmetbey University \hspace{225mm}  Karaman 70100, Turkey \hspace{225mm}    gulhanayar@kmu.edu.tr}}

\date{}

\maketitle

\begin{abstract}
In this paper, we study classification of magnetic curves corresponding to Killing vector fields of $\mathbb H^3$. First, we solve the geodesic equation analytically. Then we calculate the trajectories generated by all the six Killing vector fields, which are considered as magnetic field vectors, by using perturbation method up to first order with respect to the strength of the magnetic field. We present a comparison of our solution with the numerical solution for one case. We also prove that 3-dimensional $(\alpha)$-Kenmotsu manifolds cannot have any magnetic vector field in the direction of their Reeb vector fields.
\end{abstract}


\section{Introduction}

One of the key areas of research in differential geometry and physics is the
study of magnetic fields and the magnetic curves that correspond to them on
various manifolds. Charged particles travelling along a magnetic field
produce  magnetic curves on Riemannian manifolds.Various magnetic
fields were also extended to various ambient spaces~ \cite{comtet1987landau,sunada1993magnetic,adachi1994kahler,cabrerizo2009contact,dructua2013killing,dructua2011magnetic}. Killing magnetic
trajectories have been derived on some $3-$dimensional warped product  spaces~\cite{iqbal2020magnetic}.

A closed $2-$form defines a magnetic field on a Riemannian
manifold. This definition is derived from the fact that static magnetic fields
on a Euclidean $3-$space can be thought of as a generalization of a closed $
2-$form on a Riemannian manifold, see e.g. \cite{comtet1987landau,sunada1993magnetic}. A magnetic curve is the
path taken by a charged particle on which a magnetic field exerts a force. It is the result of solving the Lorentz
equation, a second order differential equation related to the magnetic
field. Geodesics under arclength parameterization equation is generalized by
the Lorentz equation.

A geometrical method for studying magnetic fields in
three-dimensional Riemannian manifolds has been developed where the relation between the vector fields and 2-forms was utilized \cite{Cabrerizo2}. Divergence-free vector fields define magnetic fields on three-dimensional Riemannian manifolds.

There has been a growing interest for research of magnetic curves on various geometric
structures in the last two decades. Here we highlight the most relevant works for our study which is not an exhaustive list of all magnetic curve studies. Druta and Munteanu have investigated Killing magnetic curves in Minkowski 3-space~\cite{dructua2013killing}. Cabrerizo et al. have studied the contact magnetic flow in 3D
Sasakian manifolds~\cite{cabrerizo2009contact}. Inoguchi and Munteanu have investigated
contact magnetic curves in the real special linear group of degree 2 (S$\text{L}_2\mathbb R$)~\cite{Inoguchi3}.
Jiang and Sun have studied the local differential geometrical
properties of the lightlike Killing magnetic curves in de Sitter 3-space~\cite{Jiang-Sun}. Munteanu and Nistor have provided the classification of Killing magnetic curves in $ \mathbb S^2 \times \mathbb R $~\cite{munteanu}. Ayd\i n
have classified the magnetic curves with constant curvature in a Galilean
3-space~\cite{aydin}. Magnetic curves with respect to the canonical contact
structure of the space $\text{Sol}_3$ have been investigated by Erjavec~\cite{zlatko} and
Altunba\c{s} have obtained some explicit formulas for Killing magnetic
curves in non-flat Lorentzian-Heisenberg spaces~\cite{altunbas}. Inoguchi studied some special curves in 3-dimensional hyperbolic geometry and solvgeometry~\cite{Inoguchi2}. Magnetic curves and linking numbers in the 3-sphere  ($\mathbb S^3$) and hyperbolic 3-space ($\mathbb H^3$) have been studied by De Turck and Gluck \cite{DeTurck1,DeTurck2}. However, their studies mostly deal with topological aspects of the subject and they do not provide explicit solutions for Killing magnetic curves.

The organization of the paper is as follows. In Section~\ref{sec:prelim} we give fundamental definitions that we will use in the subsequent parts of the paper. In Section~\ref{gen-struct-H3} we highlight basic properties of the $\mathbb H^3$ manifold. In Section~\ref{sec:geodesic_eom} we derive and solve the geodesic equation. In Section~\ref{sec:K_magnetics} we list six Killing vectors of the $\mathbb H^3$ manifold, which we multiply by constants $B_i$ where $i = 1,2,3,4,5,6$ in order to control the strength of the magnetic field. We have given analytical solutions for each case up to first order in $B_i$. Finally in Section~\ref{sec:conlusion} we conclude the paper. We use units where the mass ($m$) of the particle, its electrical charge ($q$) is related by $q/m = -1$.

\section{Preliminaries}\label{sec:prelim}

Let $(M,\phi ,\xi ,\eta ,g)$ be an $n$-dimensional differentiable
manifold, $(n=2m+1)$. $M$ is called an almost contact Riemannian manifold, where $\phi 
$ is a $(1,1)-$tensor field, $\xi $ is the Reeb vector field, $\eta $
is a $1-$form and $g$ is the Riemannian metric. The linear frame bundle's
structural group $\text{GL}_{m} \mathbb R$ in an almost contact manifold $M$ is reducible to 
$U(n)\times \left\{ 1\right\} $. Moreover $(\phi ,\xi ,\eta ,g)$
-structure satisfies the following conditions \cite{blair2},

\begin{equation}
    \phi^\mu_{\phantom{\mu}\alpha} \phi^\alpha_{\phantom{\alpha}\nu} = -\delta^\mu_\nu + \xi^\mu\eta_\nu, \label{2.1}
\end{equation}

and

\begin{equation}
    \eta_\mu\xi^\mu = 1, \quad \phi^\mu_{\phantom{\mu}\nu}\xi^\nu = 0,\quad \eta_\mu \phi^\mu_{\phantom{\mu}\nu} = 0. \label{2.2}
\end{equation}

Additionally, because $U(n)\subset SO(2m+1)$, $M$ admits a Riemannian
metric $g$ satisfying

\begin{equation}
    g_{\alpha\beta} \phi^\alpha_{\phantom{\alpha}\mu} \phi^\beta_{\phantom{\beta}\nu} = g_{\mu\nu} - \eta_\mu \eta_\nu. \label{2.3}
\end{equation}

A metric of this type is known as an associated metric of the almost contact manifold $(M,\phi ,\xi ,\eta ,g)$. {The $(1,1)-$tensor field $\phi $ is anti-symmetric and $\eta $ is metric dual of $\xi $ so we have} 

\begin{equation}
    \phi^\mu_{\phantom{\mu}\nu} = - \phi_\nu^{\phantom{\nu}\mu}, \quad \text{and, } g_{\mu\nu} \xi^\nu = \eta_\mu.\label{2.4}
\end{equation}

Assume $M$ is an oriented $n$-dimensional Riemannian manifold $(n\geq 2)$. A
charged particle travelling across a manifold while being affected by a
magnetic field is represented as a magnetic curve. A closed $2-$form $F$ is
a magnetic field in $(M,g)$. The $(1,1)-$tensor field $\phi $ that
corresponds to the Lorentz force of a magnetic field $F$ on $(M,g)$ is
defined by \cite{cabrerizo2009contact}:

\begin{equation}
    g_{\mu\alpha} \phi^\alpha_{\phantom{\alpha}\nu} = F_{\mu\nu}. \label{2.5}
\end{equation}

\begin{definition} \label{a-Kenmotsu}
An $\alpha$-Kenmotsu manifold is an almost contact manifold satisfying the following conditions :
\begin{itemize}
\item[(i)] $d \eta=0$.
\item[(ii)] $d F= 2 \alpha \ \eta \wedge F \  ( \alpha \in \mathbb R - \{0\})$.
\item[(iii)] The Nijenhuis tensor $ N_\phi(X,Y) $ given by the following relation vanishes for any $X,Y \in \Gamma(TM)$. 
  $$ N_\phi(X,Y) =[\phi X, \phi Y]-\phi[\phi X, Y]-\phi[ X, \phi Y]+ \phi^2 [X,Y] $$
\end{itemize} 
\end{definition}
Moreover, the following relation holds for an $\alpha$-Kenmotsu manifold for any vector field $X$ \cite{alfa-Kenmotsu}

\begin{align}
    &\nabla_X \xi= \alpha (X + h(X)-\eta(X) \xi)=\alpha (h (X)-\phi^2(X) ) \nonumber  \\ 
   & \text{ or } \   \mathrm{\nabla}_\mu\xi^\nu=\alpha \ ( {h_\mu}^\nu - {\phi_\mu}^\sigma{\phi_\sigma}^\nu )  
\end{align}

\noindent where $h$ is a trace-free (1,1)-tensor field defined as $h:=\frac{1}{2\alpha} \phi (\mathcal{L}_\xi \phi)$, $\mathcal{L}$ denoting Lie derivative operator\footnote{Note that all conditions stated above apply to ($-\alpha$)-Kenmotsu manifolds by changing $\alpha$ to $-\alpha$.}.

A curve $\gamma(t)$ on $M$ is called a magnetic curve if it satisfies the Lorentz equation:

\begin{equation}
 \mathrm{\nabla}_{\gamma^\prime}\gamma^\prime= V \times  \gamma^\prime = \phi (\gamma^\prime) 
\end{equation}
\noindent or in index notation  

\begin{equation}
\frac{\mathrm d x^\nu}{\mathrm dt} \nabla_\nu \frac{\mathrm d x^\mu}{\mathrm dt} = \phi^\mu_{\phantom{\mu}\alpha} \frac{\mathrm d x^\alpha}{\mathrm dt} \nonumber
\end{equation}

\noindent where $V$ is a vector field on $M$ associated with $F$ such that $ F_V= i_V dv_g$ ($i$ is denoting the interior product on $M$ and $dv_g$ is the volume form of $M$), and $\nabla $ is the Levi-Civita connection on $M$~\cite{Cabrerizo2}. A normal
magnetic curve is a magnetic curve whose arclength parameterization satisfies $|\gamma'(t)|=1$. If a magnetic curve $\gamma (t)$
fulfills the equation $\nabla _{\gamma'}\gamma'=0$, it is referred to as a geodesic curve. 

A vector field $K=K^\mu \partial_\mu$ on $M$ is said to be Killing vector field if it satisfies the following Killing equation which is written in covariant form.

\begin{equation} \label{Killing_general}
    \nabla_\mu K_\nu + \nabla_\nu K_\mu = 0,
\end{equation}

Magnetic fields $F$ obtained from Killing vector fields are called Killing magnetic fields and the trajectories corresponding to the Killing magnetic fields are called the Killing magnetic curves. Hence, Killing magnetic curves can be viewed as a sub-class of general magnetic curves which require only the existence of a divergence-free vector field\footnote{For instance, $B = c_1\ \cos(\alpha x) \sin(\alpha y)\partial_x-c_1\ \sin(\alpha x ) \cos(\alpha y)\partial_y+c_2\ e^{2\alpha z}\partial_z$ is a divergence-free vector field and naturally defines a magnetic field on  $\mathbb H^3$ but it does not satisfy Killing equation \eqref{Killing_general}, thus it is not a Killing vector field.}. A Killing magnetic field is defined by the closed 2-form $ F_K= i_K dv_g$ for a Killing vector field $K$ on $M$. Here the closeness of $F_K$ is guaranteed by Killing vectors being divergence-free. Then the Lorentz force $\phi_K$ corresponding to the Killing magnetic field $F_K$, and the Lorentz equation become the following

\begin{align} \label{Lorentz}
   \phi_K (X)= K \times X \ , \quad  \mathrm{\nabla}_{\gamma^\prime}\gamma^\prime= K \times  \gamma^\prime.
\end{align}

Note that $\phi_K$ in \eqref{Lorentz} is not the same $\phi$ of the original contact structure which has to be compatible with Reeb vector field $\xi$ and dual 1-form $\eta$ \cite{F_Killing}. Here the index $K$ has been used to emphasize the fact that the chosen Killing vector field defines $\phi_K$.

\section{Geometric Structure of $\mathbb H^3$}\label{gen-struct-H3}
We recall some relevant geometric properties for the hyperbolic 3-space $\mathbb H^3$ in this  section. $\mathbb H^3(-\alpha^2)$ is isometric to a solvable Lie group $G_\alpha$ \cite{KokubuH3}:

\begin{align}
G_\alpha=\left\{\left(\begin{matrix}\begin{matrix}1\\\begin{matrix}0\\\begin{matrix}0\\0\\\end{matrix}\\\end{matrix}\\\end{matrix}&\begin{matrix}0\\\begin{matrix}e^{\alpha z}\\\begin{matrix}0\\0\\\end{matrix}\\\end{matrix}\\\end{matrix}&\begin{matrix}\begin{matrix}0\\\begin{matrix}0\\\begin{matrix}e^{\alpha z}\\0\\\end{matrix}\\\end{matrix}\\\end{matrix}&\begin{matrix}z\\\begin{matrix}x\\\begin{matrix}y\\1\\\end{matrix}\\\end{matrix}\\\end{matrix}\\\end{matrix}\\\end{matrix}\right)\left| (x,y,z) \in \mathbb{R}^3\ \right.\right\} \subset \text{GL}_4\mathbb{R} \nonumber
\end{align}
equipped with left-invariant metric

\begin{equation}
    g=e^{-2\alpha z}\mathrm dx^2\ +e^{-2\alpha z}\mathrm dy^2\ +\mathrm dz^2\
\end{equation}

 It has also been shown that $\mathbb{H}^3$ can be represented by the quotient group $\text{SL}_2\mathbb{C}/\text{SU}_2 $ as a Riemannian symmetric space in \cite{DorfmeisterH3}. In this study we are using  differential geometric approach rather than group theoretic approach. We adopt the following global orthonormal frame and the corresponding dual co-frame on $\mathbb H^3(-\alpha^2)$
\begin{eqnarray}
    \eone = e^{\alpha z} \frac{\partial}{\partial x} ,\quad \etwo = e^{\alpha z} \frac{\partial}{\partial y}, \quad \ethree = \frac{\partial}{\partial z} \\ \nonumber
         \mathbf{e^1} = e^{-\alpha z} \text{d}x ,\quad   \mathbf{e^2} = e^{-\alpha z} \text{d}y , \quad   \mathbf{e^3} =\text{d}z 
\end{eqnarray}

 The Christoffel coefficients in the coordinate basis were obtained by direct calculation (lower two indices are symmetric, and we list only the nonzero terms):

\begin{align}
    \Gamma^x_{xz} = -\alpha, \quad \Gamma^y_{yz} = -\alpha, \quad \Gamma^z_{xx} = \alpha e^{-2\alpha z}, \quad \Gamma^z_{yy} = \alpha e^{-2\alpha z}. 
\end{align}

 We use this information to calculate the the following table of covariant derivatives ($i$ denotes the rows, and $j$ denotes the columns)
\begin{equation}
  \nabla_{\mathbf{e_i}}\mathbf{e_j}=\left(\begin{matrix}\alpha \ethree&0&-\alpha \eone\\0&\alpha \ethree&-\alpha \etwo\\0&0&0\\\end{matrix}\right),\ \ \ \ i,j\ \in \{1,2,3\}  
\end{equation}

\noindent and Lie brackets 

\begin{align}
   \left[\mathbf{e_i},\mathbf{e_j}\right]=\left(\begin{matrix}0&0&-\alpha \eone\\0&0&-\alpha \etwo\\\alpha \eone&\alpha \etwo&0\\\end{matrix}\right), \ \ \ \ i,j\in \{1,2,3\} 
\end{align}

 Killing vector fields of $\mathbb H^3 $ is obtained by solving \eqref{Killing_general} (see Appendix~\ref{app:KillingH3} for details)

\begin{align}\label{Killingvecs}
         &\mathbf{K_1} = \partial_x, \quad \mathbf{K_2} = \partial_y, \quad \mathbf{K_3} = y\partial_x - x\partial_y, \quad  \mathbf{K_4} = \alpha x \partial_x + \alpha y \partial_y + \partial_z, \nonumber \\
       &\mathbf{K_5} = \left( \frac{\alpha }{2}\left( x^{2}-y^{2}\right) -\frac{e^{2 \alpha z}}{2\alpha }\right) \partial _{x}+\alpha x y \partial _{y}+x\partial _{z} , \nonumber \\
    &\mathbf{K_6} = \alpha x y \partial _{x}+\left( \frac{\alpha }{2}\left(y^{2}-x^{2}\right) -\frac{e^{2 \alpha z}}{2\alpha }\right) \partial_{y}+y\partial _{z}   
\end{align}

 These vector fields form a basis of a Lie algebra of Killing vector fields whose Lie brackets are obtained as
\begin{align}
&[\mathbf{K_1},\mathbf{K_2}]=0, \ \ [\mathbf{K_1},\mathbf{K_3}]=-\mathbf{K_2}, \ \ [\mathbf{K_1},\mathbf{K_4}]= \alpha \mathbf{K_1}, \ \ [\mathbf{K_1},\mathbf{K_5}]= \mathbf{K_4}, \ \ [\mathbf{K_1},\mathbf{K_6}]= \alpha \mathbf{K_3} \nonumber \\ 
&[\mathbf{K_2},\mathbf{K_3}]=\mathbf{K_1}, \quad [\mathbf{K_2},\mathbf{K_4}]=\alpha \mathbf{K_2}, \quad [\mathbf{K_2},\mathbf{K_5}]= -\alpha \mathbf{K_3}, \quad [\mathbf{K_2},\mathbf{K_6}]= \mathbf{K_4}\nonumber \\ 
&[\mathbf{K_3},\mathbf{K_4}]=0, \quad [\mathbf{K_3},\mathbf{K_5}]= \mathbf{K_6}, \quad [\mathbf{K_3},\mathbf{K_6}]= - \mathbf{K_5} \nonumber \\ 
&[\mathbf{K_4},\mathbf{K_5}]=\alpha \mathbf{K_5}, \quad [\mathbf{K_4},\mathbf{K_6}]= \alpha \mathbf{K_6}, \quad [\mathbf{K_5},\mathbf{K_6}]= 0
\end{align}
 It is known that $\mathbb H^3(-\alpha^2)$ is a ($-\alpha$)-Kenmotsu manifold \cite{Inoguchi}. Here we summarize its almost contact structure. (1,1)-tensor $\phi$ in Section \ref{sec:prelim} can be chosen such that it satisfies the following

\begin{align}
    &\phi(\eone)=\etwo \ , \quad \phi(\etwo)=-\eone \ , \quad \phi(\ethree)=0 \quad \text{or} \nonumber \\
    &\phi(\partial_x)= \partial_y \ , \quad  \phi(\partial_y)= -\partial_x \ , \quad  \phi(\partial_z)= 0 
\end{align}
   
 Equations \eqref{2.1}-\eqref{2.4} are satisfied for $\xi=\partial_z$ and $\eta=\text{d}z$. Moreover, the coefficients of the 2-form $F_{\mu \nu}$ can be computed as a matrix by using \eqref{2.5}

\begin{align}\label{Fmunu}
    F_{\mu \nu}= \left(\begin{matrix}0&e^{-2\alpha  z}&0\\-e^{-2\alpha  z}&0&0\\0&0&0\\\end{matrix}\right)
\end{align}

Consequently, 2-form $F$ is obtained from \eqref{Fmunu} as

\begin{equation}
  F= \frac{1}{2} F_{\mu \nu} d x^\mu \wedge d x^\nu = e^{-2\alpha  z} d x \wedge d y  
\end{equation}

It is easy to check that the 2-form F and $\eta$ satisfies the following relation which is one of the conditions of  $(-\alpha)$-Kenmotsu manifolds.
\begin{equation}
 d F= -2 \alpha \ \eta \wedge F = -2 \alpha  e^{-2\alpha  z} d x \wedge d y \wedge d z  
\end{equation}

Finally, one needs to check the third condition of Definition-\ref{a-Kenmotsu} for $\mathbb H^3$. We show one example calculation for $N_\phi(e_1,e_2)$ , other combinations of basis vectors can be computed similarly. Thus, vanishing of Nijenhuis tensor ensures that $\mathbb H^3$ is a ($-\alpha$)-Kenmotsu manifold.
\begin{align}
    N_\phi(e_1,e_2)&=[\phi (e_1), \phi (e_2)]-\phi[\phi (e_1), e_2]-\phi[ e_1, \phi (e_2)]+ \phi[\phi[e_1,e_2]]  \nonumber \\
    &=[e_2,-e_1]-\phi[e_2, e_2]-\phi[e_1,-e_1]+\phi[\phi(0)]=0
\end{align}

 A remark should be made here about special case of magnetic curves on $\mathbb H^3$. Unlike other almost contact manifolds, fundamental 2-form $F$ coming from the contact structure is not closed on $\mathbb H^3$ hence it does not correspond to a magnetic field. In addition, Reeb vector field $\xi$ of $\mathbb H^3$ is not Killing. Similar situation appears in $Sol_3$ space, its Reeb vector field is not Killing but divergence-free ($\nabla \cdot \xi =0 $) and there exists a closed 2-form corresponding to a magnetic field for the Reeb vector of $Sol_3$ space \cite{Erjavec2}. We construct and utilize the following lemma for $\alpha$-Kenmotsu manifolds. 

\begin{lemma} \label{lemma3.1}
Let $(M,\phi,\xi,\eta,g) $ be a 3-dimensional $\alpha$-Kenmotsu manifold. Then there does not exist any magnetic curve $\gamma$ associated with the Reeb vector field $\xi$ of the $\alpha$-Kenmotsu manifold.
\end{lemma}

\begin{proof}
For any $\alpha$-Kenmotsu manifold we have the following relations. 
$$ \ {{\phi_\beta}^\mu\ \phi_\mu}^\sigma=-{\delta\ }_\beta^\sigma+\xi^\sigma\eta_\beta   \ , \quad  \mathrm{\nabla}_\mu\xi^\nu=\alpha \ ( {h_\mu}^\nu - {\phi_\mu}^\sigma{\phi_\sigma}^\nu )$$ 
$$ \text{Tr}(\mathrm{\nabla}_\mu\xi^\nu)=\alpha \ \text{Tr}({h_\mu}^\nu - {\phi_\mu}^\sigma{\phi_\sigma}^\nu)  \Rightarrow \nabla_\mu\xi^\mu=-\alpha\ {\phi_\mu}^\sigma{\phi_\sigma}^\mu\ $$
$$ \text{Tr}({\phi_\mu}^\sigma{\phi_\beta}^\mu)=\text{Tr}(-{\delta\ }_\beta^\sigma+\xi^\sigma\eta_\beta\ )  \ \Rightarrow{\phi_\mu}^\sigma{\phi_\sigma}^\mu=-{\delta\ }_\sigma^\sigma+\xi^\sigma\eta_\sigma=-3+1=-2 $$
Thus, 
$$\nabla_\mu\xi^\mu=-\alpha\ {\phi_\mu}^\sigma{\phi_\sigma}^\mu=2\ \alpha  \Rightarrow\ \nabla \cdot \xi \neq 0$$
By definition, a magnetic curve requires the existence of a divergence free vector field in 3-dimensions. Since Reeb vector field of $\alpha$-Kenmotsu manifold $\xi$ will always have a non-vanishing divergence, a magnetic curve can not exist associated with $\xi . \ \ $ 
\end{proof}

\begin{corollary}
Lemma~\ref{lemma3.1} also applies to $(-\alpha)$-Kenmotsu manifolds with $\alpha \rightarrow -\alpha $ which gives $\nabla \cdot \xi=-2 \alpha$. There does not exist any magnetic curve $\gamma$ associated with the Reeb vector field $\xi$ of $\mathbb H^3$ since it is a 3-dimensional $(-\alpha)$-Kenmotsu manifold.
\end{corollary}

We also note that Killing vector fields of $\mathbb H^3$ are not compatible with the contact structure of $(-\alpha)$-Kenmotsu manifolds. Their metric duals are always non-closed (d$\eta_K \neq 0$)  which violates one of the conditions of $(-\alpha)$-Kenmotsu manifolds. However, this situation poses no obstruction in obtaining Killing magnetic curves which is only related with manifold’s intrinsic geometric structure and solution of Lorentz equation.

\section{The Geodesic Equation and its Solution\label{sec:geodesic_eom}}

Here we define the geodesic equation in Subsection~\ref{sec:deriv_geodesic_eqn} then solve it in Subsection~\ref{sec:geodesic_eqn_solution}.

\subsection{Derivation of the Geodesic Equation\label{sec:deriv_geodesic_eqn}}

Let $\gamma'(t)$ be a curve in $\mathbb H^3$. If it satisfies the equation, $\nabla_{\gamma'}\gamma' = 0$, it is said to be a \emph{geodesic curve}.  In the orthonormal basis, $\{\eone,\etwo,\ethree\}$ we can write $\gamma' = \dot\gamma^1 \eone + \dot\gamma^2 \etwo + \dot\gamma^3 \ethree$ (a `dot' means derivative with respect to time) and the geodesic equation becomes:

\begin{align}
    \nabla_{\gamma'} \gamma' &= (\gamma' \cdot \nabla) \gamma',\\
    &= \dot\gamma^i \nabla_{\mathbf{e_i}} (\dot\gamma^j \mathbf{e_j}),\\
    &= \dot\gamma^i (\nabla_{\mathbf{e_i}} \dot\gamma^j)\mathbf{e_j} + \dot\gamma^i \dot\gamma^j \nabla_{\mathbf{e_i}} \mathbf{e_j},\\
    &= \partial_t \dot\gamma^j \mathbf{e_j} + \dot\gamma^i \dot\gamma^j \nabla_{\mathbf{e_i}} \mathbf{e_j},\\
    &= 0.
\end{align}

When the covariant derivatives ($\nabla_{\mathbf{e_i}} \mathbf{e_j}$) are taken into account, we find the following three equations:

\begin{align}
    \ddot\gamma^1 - \alpha \dot\gamma^1 \dot\gamma^3 &= 0,\label{eq:geodesic_Gamma_1_eqn}\\
    \ddot\gamma^2 - \alpha \dot\gamma^2 \dot\gamma^3 &= 0,\label{eq:geodesic_Gamma_2_eqn}\\
    \ddot \gamma^3 + \alpha \left(\dot\gamma^1\right)^2 +  \alpha \left(\dot\gamma^2\right)^2 &= 0. \label{eq:geodesic_Gamma_3_eqn}
\end{align}

\subsection{Solution of the Geodesic Equation\label{sec:geodesic_eqn_solution}}

The orthonormal basis has been useful for obtaining the geodesic equation. The relation between the components of $\gamma'$ in the orthonormal basis ($\dot\gamma^i$) and the coordinate basis ($\dot x^i$) is simply: $\dot\gamma^1 = \dot x e^{-\alpha z}, \dot\gamma^2 = \dot y e^{-\alpha z}, \dot\gamma^3 = \dot z$. Let us re-write the geodesic equations of motion, Equations~(\ref{eq:geodesic_Gamma_1_eqn},\ref{eq:geodesic_Gamma_2_eqn},\ref{eq:geodesic_Gamma_3_eqn}), using the coordinate basis (where $\gamma' = \dot x \partial_x + \dot y \partial_y + \dot z \partial_z$) and write these equations as:

\begin{align}
    \ddot x - 2\alpha \dot x \dot z &= 0,\\
    \ddot y - 2\alpha \dot y \dot z &= 0,\\
    \ddot z + \alpha (\dot x^2 + \dot y^2) e^{-2\alpha z} &= 0. \label{eq:geodesic_ddot_z}
\end{align}

The first two equations are of first order in $\dot x$ and $\dot y$ and can be easily integrated:

\begin{align}
    \dot x &= c_1 e^{2 \alpha z},\label{eq:geodesic_dot_x}\\
    \dot y &= c_2 e^{2 \alpha z},\label{eq:geodesic_dot_y}
\end{align}

\noindent where $c_1, c_2$ are constants of integration. Physically, these constants are proportional to initial velocity in $x$ and $y$ directions. For example, if the initial time is $t_0$, $c_1 = v_{0x} e^{-2\alpha z(t_0)}$ and $c_2 = v_{0y} e^{-2\alpha z(t_0)}$. {$c_1, c_2$ are linearly related to velocities in $x,y$ directions on the $z = z(t_0)$ plane. Actually, since the differential equations are autonomous in $t$, instead of $t_0$ one may choose any other $t_1$ in order to specify the velocities. As we shall see later, this is the reason why we kept the constant $c_4$ (which appears in the form of $t - c_4$) in our solutions.} When we put these results in Equation~(\ref{eq:geodesic_ddot_z}) we obtain:

\begin{align}
    \ddot z + \alpha (c_1^2 + c_2^2) e^{2\alpha z} = 0.
\end{align}

We multiply this equation with $2 \dot z$ and obtain:

\begin{align}
    \partial_t(\dot z^2) + (c_1^2 + c_2^2) \partial_t e^{2\alpha z} &= 0,\\
    \intertext{we integrate with respect to time and obtain:}
    \dot z^2 + (c_1^2 + c_2^2) e^{2\alpha z} &= c_3^2,\label{eq:geodesic_dot_z_2}
\end{align}

\noindent where $c_3 > 0$ or equals to zero if the particle stands still, that is $\dot x = \dot y = \dot z = 0$. {The right hand side of (\ref{eq:geodesic_dot_z_2}) must be non-negative because the expression on the left hand side is non-negative. This is why we chose to denote the right hand side as $c_3^2$. As a matter of convention we chose $c_3 \geq 0$. The solutions of the differential equations are invariant under $c_3 \leftrightarrow -c_3$ as can be seen in Equations~(\ref{eq:geodesic_exp_2az}-\ref{eq:geodesic_y}). Moreover, $c_3$ is related to the initial velocities through $c_3^2 = v_{0z}^2 + (v_{0x}^2 + v_{0y}^2) e^{-2\alpha z(t_0)}$}. {In essence $c_3$ is the speed (in $\mathbb H^3$) of the geodesic curve which can be obtained through $\mathrm ds^2 / \mathrm dt^2$ where $\mathrm ds^2$ is given by the $\mathbb H^3$ metric. It is seen that $|\gamma'|^2 = c_3^2$ and $c_3$ should be constant in order $\gamma'$ to define a geodesic curve. In order to see that apply $\nabla$ on both sides and obtain $2 \nabla_{\gamma'} \gamma' = \nabla c_3^2$. By scaling the time parameter it is possible to map $c_3 \to 1$, however we favored the option to display each constant of integration. This is the reason why we kept the $c_3$ dependence explicit.} Since we are interested in \emph{motion}, from now on we suppose $c_3 > 0$. We can easily integrate Equation~(\ref{eq:geodesic_dot_z_2}) and obtain ($c_4$ is a constant of integration):

\begin{align}
    t - c_4 &= \int \mathrm dz \left(c_3^2 - (c_1^2 + c_2^2) e^{2\alpha z}\right)^{-1/2} \nonumber \\
    &= \frac{1}{c_3} \int \mathrm dz \left(1 - A e^{2\alpha z}\right)^{-1/2}.\\
    \intertext{where $A = (c_1^2 + c_2^2)/c_3^2$. Let $u = A e^{2\alpha z}$, and obtain:}
    &= \frac{1}{2\alpha c_3} \int \frac{\mathrm du}{u\sqrt{1-u}}.
    \intertext{Evaluating the integral yields \cite{mathematica}:}
    &= -\frac{1}{\alpha c_3} \tanh^{-1}(\sqrt{1-u}).
\end{align}

By inverting this equation and using the definition of $u$ in terms of $z$, we obtain:

\begin{align}\label{eq:geodesic_exp_2az}
    e^{2\alpha z} &= \frac{c_3^2}{c_1^2 + c_2^2} \left[1 - \tanh^2(\alpha c_3 (t - c_4))\right] = \frac{c_3^2}{c_1^2 + c_2^2} \sech^2(\alpha c_3 (t - c_4)) \nonumber \\ 
    z(t)&=\frac{1}{2\alpha} \text{Log}\left(\frac{{c_3}^2}{({c_1}^2+{c_2}^2)} \sech^2(\alpha c_3(t-c_4))\right)
\end{align}

Since we know $z$ in terms of time, we can use this information to integrate $\dot x$ and $\dot y$ in Equation~(\ref{eq:geodesic_dot_x}) and in Equation~(\ref{eq:geodesic_dot_y}):

\begin{align}
    x(t) &= \frac{1}{\alpha} \frac{c_1 c_3}{c_1^2 + c_2^2} \tanh(\alpha c_3 (t - c_4)) + c_5,\label{eq:geodesic_x}\\
    y(t) &= \frac{1}{\alpha} \frac{c_2 c_3}{c_1^2 + c_2^2} \tanh(\alpha c_3 (t - c_4)) + c_6,\label{eq:geodesic_y}
\end{align}

\noindent where $c_5, c_6$ are constants of integration. In essence, Equations~(\ref{eq:geodesic_exp_2az},\ref{eq:geodesic_x},\ref{eq:geodesic_y}) gives the equations that define a geodesic curve. As it should be, there are six constants of integration.

\section{Killing magnetic curves in $\mathbb H^3$}\label{sec:K_magnetics}

As we have given in Section~\ref{gen-struct-H3} there are six Killing vectors in $\mathbb H^3$ 

\begin{align}
    K_{(1)} &= B_1 \partial_x,\\
    K_{(2)} &= B_2 \partial_y,\\
    K_{(3)} &= B_3 (y\partial_x - x\partial_y),\\
    K_{(4)} &= B_4 (\alpha x \partial_x + \alpha y \partial_y + \partial_z),\\
    K_{(5)} &= B_5\left( \left( \frac{\alpha }{2}\left( x^{2}-y^{2}\right) -\frac{e^{2 \alpha z}}{2\alpha }\right) \partial _{x}+\alpha x y \partial _{y}+x\partial _{z}  \right), \\
    K_{(6)} &= B_6\left(  \alpha x y \partial _{x}+\left( \frac{\alpha }{2}\left(y^{2}-x^{2}\right) -\frac{e^{2 \alpha z}}{2\alpha }\right) \partial_{y}+y\partial _{z}   \right).
\end{align}
    
where we put the constants $B_i$ to control the \emph{strength} of the magnetic field (note that the \emph{units} of each $B_i$ may vary.). We do this because we could not find analytic solutions for arbitrary $B_i$ and will give equations up to first order in $B_i$. However in all cases, we will give the full equations for arbitrary $B_i$ and then do perturbation analysis.

The Killing magnetic curve generated by the $i^{\text{th}}$ Killing vector field is given by Lorentz equation \eqref{Lorentz}:

\begin{align}
    \nabla_{\gamma'}\gamma' = K_{(i)} \times \gamma',\label{eq:kmc_def}
\end{align}

where the vector product is calculated using the orthonormal basis in a tangent space of the underlying manifold. If we define $F_{(i)} \equiv K_{(i)} \times \gamma'$, we obtain the following expressions (remember $\gamma' = \dot\gamma^i \mathbf{e_i}$):

\begin{align}
    \frac{F_{(1)}}{B_1} &= e^{-\alpha z} (\dot\gamma^2 \ethree - \dot\gamma^3 \etwo),\label{eq:kmc_F_1}\\
    \frac{F_{(2)}}{B_2} &= e^{-\alpha z} (\dot\gamma^3 \eone - \dot\gamma^1 \ethree),\label{eq:kmc_F_2}\\
    \frac{F_{(3)}}{B_3} &= e^{-\alpha z} (-x\dot\gamma^3\eone - y\dot\gamma^3\etwo + (x\dot\gamma^1 + y\dot\gamma^2)\ethree),\label{eq:kmc_F_3}\\
    \frac{F_{(4)}}{B_4} &= (\alpha y \dot\gamma^3 e^{-\alpha z} - \dot\gamma^2)\eone - (\alpha x \dot\gamma^3e^{-\alpha z}-\dot\gamma^1)\etwo + \alpha e^{-\alpha z}(x\dot\gamma^2 - y\dot\gamma^1)\ethree.\label{eq:kmc_F_4}
\end{align}

\begin{eqnarray}
        \frac{F_{(5)}}{B_5} &=& x (\alpha y \dot\gamma^3 e^{-\alpha z} - \dot\gamma^2)\eone +  \left( x \dot\gamma^1 + \frac{e^{\alpha z}}{2 \alpha} \dot\gamma^3+  \frac{e^{-\alpha z}\alpha}{2}  (y^2-x^2)  \dot\gamma^3 \right)\etwo \nonumber \\
  &&+  \left(\frac{e^{-\alpha z} \alpha}{2} ((x^2-y^2) \dot\gamma^2-2 x y \dot\gamma^1) - \frac{e^{\alpha z}}{2 \alpha} \dot\gamma^2 \right)\ethree, \label{eq:kmc_F_5}\\
     \frac{F_{(6)}}{B_6} &=& \left( -y \dot\gamma^2 - \frac{e^{\alpha z}}{2 \alpha} \dot\gamma^3+  \frac{e^{-\alpha z}\alpha}{2}  (y^2-x^2)  \dot\gamma^3 \right)\eone + y (\dot\gamma^1- \alpha x \dot\gamma^3 e^{-\alpha z} )\etwo \nonumber \\
   &&+\left(\frac{e^{-\alpha z} \alpha}{2} ((x^2-y^2) \dot\gamma^1+2 x y \dot\gamma^2) + \frac{e^{\alpha z}}{2 \alpha} \dot\gamma^1 \right)\ethree    \label{eq:kmc_F_6}
\end{eqnarray}  

\subsection{Magnetic Trajectory by the First Killing Vector Field}
Using the expression found for $F_{(1)}$ in Equation~(\ref{eq:kmc_F_1}) and Equation~(\ref{eq:kmc_def}), then turning them into the coordinate basis we obtain the following set of equations:

\begin{align}
    \ddot x - 2\alpha \dot x \dot z &= 0,\\
    \ddot y - 2\alpha \dot y \dot z &= -B_1 \dot z,\\
    \ddot z + \alpha(\dot x^2 + \dot y^2) e^{-2\alpha z} &= B_1 \dot y e^{-2\alpha z}.\label{eq:K_1_ddot_z}
\end{align}

The first and second equations are easily integrated (the results are exact):

\begin{align}
    \dot x &= c_1 e^{2\alpha z},\label{eq:K_1_dot_x}\\
    \dot y &= c_2 e^{2\alpha z} + \frac{B_1}{2\alpha}.\label{eq:K_1_dot_y}
\end{align}

Using these two expressions in Equation~(\ref{eq:K_1_ddot_z}) we obtain:

\begin{equation}
    \ddot z + \alpha(c_1^2 + c_2^2) e^{2\alpha z} - \frac{B_1^2}{4\alpha}e^{-2\alpha z} = 0.
\end{equation}

Multiply with $2\dot z$ then integrate to obtain (the result is exact):

\begin{align}
    \dot z^2 + (c_1^2 + c_2^2) e^{2\alpha z} + \frac{B_1^2}{4\alpha^2} e^{-2\alpha z} &= c_3^2,\label{eq:K_1_dot_z_2}
\end{align}

where $c_3 \geq 0$. In order to determine the Killing magnetic trajectories generated by $B_1 K_{(1)}$ one needs to solve these equations that are first order in $\dot x, \dot y, \dot z$. Since we could not come up with analytical solution, we will do a perturbation analysis up to first order in $B_1$. For that purpose we define $x = x_0 + B_1 x_1, y = y_0 + B_1 y_1, z = z_0 + B_1 z_1$. Seemingly the zeroth order functions $x_0, y_0, z_0$ are the ones that satisfies the geodesic equation of motion which we provided in Section~\ref{sec:geodesic_eom}. After we expand the Equations~(\ref{eq:K_1_dot_x},\ref{eq:K_1_dot_y},\ref{eq:K_1_dot_z_2}) up to first order in $B_1$ we obtain the following differential equations:

\begin{align}
    \dot x_1 &= 2\alpha c_1 e^{2\alpha z_0}z_1,\\
    \dot y_1 &= 2\alpha c_2 e^{2\alpha z_0}z_1 + \frac{1}{2\alpha},\\
    \frac{\dot z_1}{z_1} &= -\alpha(c_1^2 + c_2^2) \frac{e^{2\alpha z_0}}{\dot z_0}.
\end{align}

By integrating the last equation, we obtain:

\begin{align}
    z_1 = c_7 \tanh(\alpha c_3 (t-c_4)),
\end{align}

where $c_7$ is a new constant of integration. By using this expression in above Equations for $\dot x_1$ and $\dot y_1$ we obtain (via Mathematica~\cite{mathematica}):

\begin{align}
    x_1 &= - \frac{c_1 c_3 c_7}{c_1^2 + c_2^2} \sech^2(\alpha c_3(t-c_4)) + c_8\\
    y_1 &= - \frac{c_2 c_3 c_7}{c_1^2 + c_2^2} \sech^2(\alpha c_3(t-c_4)) + \frac{t - c_4}{2\alpha} + c_9,
\end{align}

where $c_8, c_9$ are new constants of integration. We have found the zeroth order functions in Section~\ref{sec:geodesic_eqn_solution}, so the first order solution to magnetic trajectory produced by the first Killing vector is:

\begin{align}
    x &= x_0 - B_1 \frac{c_1 c_3 c_7}{c_1^2 + c_2^2} \sech^2(\alpha c_3(t-c_4)) + \mathcal{O}(B_1^2)\\
    y &= y_0  - B_1 \frac{c_2 c_3 c_7}{c_1^2 + c_2^2} \sech^2(\alpha c_3(t-c_4)) + \frac{B_1(t - c_4)}{2\alpha} + \mathcal{O}(B_1^2)\\
    e^{2\alpha z} &= e^{2\alpha z_0} (1 + 2\alpha B_1 c_7 \tanh(\alpha c_3 (t-c_4))) + \mathcal{O}(B_1)^2.
\end{align}

where we have absorbed $c_8, c_9$ into the constants $c_5,c_6$ that are found in $x_0,y_0$ respectively.

\subsection{Magnetic Trajectory by the Second Killing Vector Field}\label{sec:K_2}
Using the expression found for $F_{(2)}$ in Equation~(\ref{eq:kmc_F_2}) and Equation~(\ref{eq:kmc_def}), then turning them into the coordinate basis we obtain the following set of equations:

\begin{align}
    \ddot x - 2\alpha \dot x\dot z &= B_2 \dot z,\\
    \ddot y - 2\alpha \dot y \dot z &= 0,\\
    \ddot z + \alpha(\dot x^2 + \dot y^2)e^{-2\alpha z} &= -B_2 \dot x e^{-2\alpha z}.
\end{align}

We will not explicitly deal with this case, because by mapping $x\to y, y\to x, B_1\to -B_2$ in the case in the previous Section, we can obtain the solutions for this case, due to a symmetry between the equations of motion in two cases. Hence, we leave the reader to obtain the full (for arbitrary $B_2$) equations of motion and we only write down the first order solutions:

\begin{align}
    x &= x_0 + B_2 \frac{c_1 c_3 c_7}{c_1^2 + c_2^2} \sech^2(\alpha c_3(t-c_4)) - \frac{B_2(t - c_4)}{2\alpha} + \mathcal{O}(B_2^2)\\
    y &= y_0  + B_2 \frac{c_2 c_3 c_7}{c_1^2 + c_2^2} \sech^2(\alpha c_3(t-c_4)) + \mathcal{O}(B_2^2)\\
    e^{2\alpha z} &= e^{2\alpha z_0} (1 - 2\alpha B_2 c_7 \tanh(\alpha c_3 (t-c_4))) + \mathcal{O}(B_2)^2.
\end{align}

\subsection{Magnetic Trajectory by the Third Killing Vector Field}
Using the expression found for $F_{(3)}$ in Equation~(\ref{eq:kmc_F_3}) and Equation~(\ref{eq:kmc_def}), then turning them into the coordinate basis we obtain the following set of equations:

\begin{align}
    \ddot x - 2\alpha \dot x\dot z &= -B_3 x \dot z,\\
    \ddot y - 2\alpha \dot y \dot z &= -B_3 y \dot z,\\
    \ddot z + \alpha(\dot x^2 + \dot y^2)e^{-2\alpha z} &= B_3 (x\dot x + y\dot y) e^{-2\alpha z}.
\end{align}

We will solve these equations of motion using a perturbation analysis. We write down $x = x_0 + B_3 x_1, y = y_0 + B_3 y_1, z = z_0 + B_3 z_1$ up to first order in $B_3$ and note that the functions with zero indices denote solutions of the geodesic equation of motion that we have found in Section~\ref{sec:geodesic_eom}. The differential equations for $x_1,y_1,z_1$ are then:

\begin{align}
    \ddot x_1 - 2\alpha \dot z_0 \dot x_1 - 2\alpha \dot x_0 \dot z_1 &= -x_0 \dot z_0,\label{eq:K_3_ddot_x1}\\
    \ddot y_1 - 2\alpha \dot z_0 \dot y_1 - 2\alpha \dot y_0 \dot z_1 &= -y_0 \dot z_0,\\
    e^{2\alpha z_0} \ddot z_1 - 2\alpha^2 (\dot x_0^2 + \dot y_0^2) z_1 + 2\alpha(\dot x_0 \dot x_1 + \dot y_0 \dot y_1) &= x_0 \dot x_0 + y_0 \dot y_0.\label{eq:K_3_ddot_z1}
\end{align}

Let us integrate the first two Equations. For that purpose we multiply Equation~(\ref{eq:K_3_ddot_x1}) by $e^{-2\alpha z_0}$ and obtain:

\begin{align}
    e^{-2\alpha z_0} \ddot x_1 - 2\alpha \dot z_0 e^{-2\alpha z_0} \dot x_1 - 2\alpha c_1 \dot z_1 = -x_0 \dot z_0 e^{-2\alpha z_0},
\end{align}

and we can write this equation as:

\begin{align}
    \partial_t (e^{-2\alpha z_0} \dot x_1 - 2\alpha c_1 z_1) = \frac{1}{2\alpha} x_0 \partial_t e^{-2\alpha z_0}.
\end{align}

We easily integrate the left hand side of the equation, and use integration by parts on the right hand side, the obtain:

\begin{align}
    \dot x_1 = 2\alpha c_1 e^{2\alpha z_0} z_1 + \frac{x_0}{2\alpha} - \frac{c_1(t-c_4)}{2\alpha}e^{2\alpha z_0} + c_7e^{2\alpha z_0}.\label{eq:K_3_x1_dot}
\end{align}

The case for the $\dot y_1$ is found by using the symmetry $x \leftrightarrow y$, and we can immediately write down:

\begin{align}
    \dot y_1 = 2\alpha c_2 e^{2\alpha z_0} z_1 + \frac{y_0}{2\alpha} - \frac{c_2(t-c_4)}{2\alpha}e^{2\alpha z_0} + c_8e^{2\alpha z_0},\label{eq:K_3_y1_dot}
\end{align}

where $c_7, c_8$ are constants of integration. We put these functions into the Equation~(\ref{eq:K_3_ddot_z1}) and obtain the following expression:

\begin{align}
    e^{-2\alpha z_0} \ddot z_1 + 2\alpha^2(c_1^2 + c_2^2) z_1 = (c_1^2 + c_2^2)(t-c_4) - 2\alpha (c_1 c_7 + c_2 c_8).
\end{align}

Let us define $\tau = \alpha c_3 (t-c_4)$ and rewrite the above equation as, after some algebraic manipulations:

\begin{align}
    \cosh^2(\tau)\partial_\tau^2 z_1+2z_1 = \frac{\tau}{\alpha^3 c_3} - \frac{2(c_1 c_7 + c_2 c_8)}{\alpha(c_1^2 + c_2^2)},
\end{align}

hence we can write:

\begin{align}
    z_1 = z_{\text{1h}} + \frac{\tau}{2\alpha^3 c_3} - \frac{c_1 c_7 + c_2 c_8}{\alpha(c_1^2 + c_2^2)},
\end{align}

where $z_{\text{1h}}$ satisfies the homogeneous equation. We find the homogeneous solution, with the help of Mathematica \cite{mathematica} to be:

\begin{align}
    z_{\text{1h}} = c_9 \tanh(\tau) + c_{10}(1 - \tau\tanh(\tau)),
\end{align}

So we can write:

\begin{align}
    z_1 = \frac{\tau}{2\alpha^3 c_3} - \frac{c_1 c_7 + c_2 c_8}{\alpha(c_1^2 + c_2^2)} + c_9 \tanh(\tau) + c_{10}(1 - \tau\tanh(\tau)).\label{eq:K_3_z1}
\end{align}

We use this information in Equation~(\ref{eq:K_3_x1_dot}) to find $x_1$ by performing integration by Mathematica \cite{mathematica}, and after a few algebraic manipulations we obtain:

\begin{multline}
    x_1 = -\frac{c_1 \tau^2}{4 \alpha ^3 c_3^2}-\frac{c_1 \log (\cosh (\tau))}{2
   \alpha ^3 \left(c_1^2+c_2^2\right)}\\
   +\tau \left(\frac{c_5}{2
   \alpha ^2 c_3}+\frac{c_1^3+c_2^2 c_1}{\alpha ^3
   \left(c_1^2+c_2^2\right){}^2} \tanh (\tau)+\frac{c_1 c_3
   c_{10} \text{sech}^2(\tau)}{c_1^2+c_2^2}\right)\\
   +\frac{1}{\alpha(c_1^2+c_2^2)}\left(\alpha c_1 c_3 c_{10}(c_1^2 + c_2^2)+\alpha c_3 c_7 - 2 c_1 c_3 (c_1 c_7 + c_2 c_8)\right)\tanh(\tau)\\
   -\frac{c_1
   c_3 c_9 \text{sech}^2(\tau)}{c_1^2+c_2^2} + c_{11},\label{eq:K_3_x1}
\end{multline}

where $c_{11}$ is a new integration constant. By using the symmetry between $x_1,y_1$ we can immediately write down the result for $y_1$ as (here $c_{12}$ is a new constant of integration):

\begin{multline}
    y_1 = -\frac{c_2 \tau^2}{4 \alpha ^3 c_3^2}-\frac{c_2 \log (\cosh (\tau))}{2
   \alpha ^3 \left(c_1^2+c_2^2\right)}\\
   +\tau \left(\frac{c_6}{2
   \alpha ^2 c_3}+\frac{c_2^3+c_1^2 c_2}{\alpha ^3
   \left(c_1^2+c_2^2\right){}^2} \tanh (\tau)+\frac{c_2 c_3
   c_{10} \text{sech}^2(\tau)}{c_1^2+c_2^2}\right)\\
   +\frac{1}{\alpha(c_1^2+c_2^2)}\left(\alpha c_2 c_3 c_{10}(c_1^2 + c_2^2)+\alpha c_3 c_8 - 2 c_2 c_3 (c_1 c_7 + c_2 c_8)\right)\tanh(\tau)\\
   -\frac{c_1
   c_2 c_9 \text{sech}^2(\tau)}{c_1^2+c_2^2} + c_{12}.\label{eq:K_3_y1}
\end{multline}

To summarize, the solution for the Killing magnetic curve, up to first order in $B_3$ is as follows:

\begin{align}
    x &= x_0 + B_3 x_1 + \mathcal{O}(B_3^2),\\
    y &= y_0 + B_3 y_1 + \mathcal{O}(B_3^2),\\
    z &= z_0 + B_3 z_1 + \mathcal{O}(B_3^2).
\end{align}

where $x_0, y_0, z_0$ are solutions for the geodesic equation and $x_1,y_1,z_1$ are given in Equations~(\ref{eq:K_3_x1},\ref{eq:K_3_y1},\ref{eq:K_3_z1}) respectively. Last but not least, wihle making use of $x_1,y_1,z_1$, remember that $\tau = \alpha c_3 (t-c_4)$.

\subsection{Magnetic Trajectory by the Fourth Killing Vector Field\label{sec:K_4}}
Using the expression found for $F_{(4)}$ in Equation~(\ref{eq:kmc_F_4}) and Equation~(\ref{eq:kmc_def}), then turning them into the coordinate basis we obtain the following set of equations:

\begin{align}
    \ddot x - 2\alpha \dot x\dot z &= \alpha B_4 y\dot z - B_4 \dot y,\\
    \ddot y - 2\alpha \dot y \dot z &= -\alpha B_4 x\dot z + B_4 \dot x,\\
    \ddot z + \alpha(\dot x^2 + \dot y^2)e^{-2\alpha z} &=\alpha B_4 (x\dot y - \dot x y) e^{-2\alpha z}.\label{eq:K_4_z_ddot}
\end{align}

We will solve these equations of motion using a perturbation analysis. We write down $x = x_0 + B_4 x_1, y = y_0 + B_4 y_1, z = z_0 + B_4 z_1$ upto first order in $B_3$ and note that the function with zero indices are solutions of the geodesic equation of motion that we have found in Section~\ref{sec:geodesic_eom}. For $x_1$ we can write the following (by keeping the first order terms in $B_4$ in the first differential equation above):

\begin{align}
    \ddot x_1 - 2\alpha \dot z_0 \dot x_1 - 2\alpha \dot x_0 \dot z_1 &= \alpha y_0 \dot z_0 - \dot y_0.
\end{align}

By multiplying with $e^{-2\alpha z_0}$, we obtain:

\begin{align}
    \partial_t(\dot x_1 e^{-2\alpha z_0}) &= 2\alpha c_1 \dot z_1 - c_2 + \alpha \dot z_0 y_0 e^{-2\alpha z_0}, \nonumber \\
    &= 2\alpha c_1 \dot z_1 -\frac{1}{2} y_0 \partial_t e^{-2\alpha z_0},
    \intertext{We integrate both sides and use integration by parts on the right hand side and obtain:}
    \dot x_1 e^{-2\alpha z_0} &= 2\alpha c_1 z_1 - c_2 t - \frac{1}{2}y_0 e^{-2\alpha z_0} + \frac{c_2 t}{2} + c_7,
    \intertext{where $c_7$ is a constant of integration. After a few algebraic manipulations, we obtain:}
    \dot x_1 &= 2\alpha c_1 z_1 e^{2\alpha z_0} - \frac{1}{2} c_2 (t-c_4) e^{2\alpha z_0} - \frac{y_0}{2} + c_7 e^{2\alpha z_0}.\label{eq:K_4_x1_dot}
    \intertext{Here we absorbed a multiple of $c_4$ in $c_7$ in order to make the change of variables easier later on. The case for the calculation of $\dot y_1$ follows along similar lines, and the results reads:}
    \dot y_1 &= 2\alpha c_2 z_1 e^{2\alpha z_0} + \frac 12 c_1 (t-c_4) e^{2\alpha z_0} + \frac{x_0}{2} + c_8 e^{2\alpha z_0},\label{eq:K_4_y1_dot}
\end{align}

where $c_8$ is a constant of integration. Using the information we have obtained through Equations~(\ref{eq:K_4_x1_dot},\ref{eq:K_4_y1_dot}) in Equation~(\ref{eq:K_4_z_ddot}) and by regarding the first order terms in $B_4$ only, we obtain:

\begin{align}
    e^{2\alpha z_0} \ddot z_1 - 2\alpha^2(\dot x_0^2 + \dot y_0^2)z_1 = -2\alpha(\dot x_0 \dot x_1 + \dot y_0 \dot y_1) + \alpha(x_0 \dot y_0 - \dot x_0 y_0.)
\end{align}

When we put the values of $x_0, \dot x_1, y_0, \dot y_1$ we obtain:

\begin{align}
    e^{-2\alpha z_0} \ddot z_1 + 2\alpha^2(c_1^2 + c_2^2) z_1 = -2\alpha (c_1 c_7 + c_2 c_8).\label{eq:K_4_z1_ddot}
\end{align}

The particular solution is easy to find:

\begin{align}
    z_{1\text{p}} = -\frac{c_1 c_7 + c_2 c_8}{\alpha (c_1^2 + c_2^2)}.
\end{align}

On the other hand, the homogenous solution $z_{1\text{h}}$ satisfies:

\begin{align}
    e^{-2\alpha z_0} \ddot z_{1\text{h}} + 2\alpha^2(c_1^2 + c_2^2) z_{1\text{h}} = 0,
\end{align}

and when we put the value of $e^{-2\alpha z_0}$ we obtain:

\begin{align}
    \cosh^2(\alpha c_3 (t - c_4)) \ddot z_{1\text{h}} + 2\alpha^2 c_3^2 z_{1\text{h}} = 0.
\end{align}

With the help of Mathematica \cite{mathematica} we find the homogenous solution as:

\begin{align}
    z_{1\text{h}} = -c_{10} + (c_9 + \alpha c_3 c_{10} (t - c_4)) \tanh(\alpha c_3 (t - c_4)),\label{eq:K_4_z1h}
\end{align}

where $c_{9}, c_{10}$ are new constants of integration. All in all, the solution we have found for $z_1$ reads as:

\begin{align}
    z_1 = -c_{10} + (c_9 + \alpha c_3 c_{10}  (t - c_4)) \tanh(\alpha c_3 (t - c_4)) -\frac{c_1 c_7 + c_2 c_8}{\alpha (c_1^2 + c_2^2)}.\label{eq:K_4_z1}
\end{align}

We use the obtained expression for $z_1$ in Equations~(\ref{eq:K_4_x1_dot},\ref{eq:K_4_y1_dot}). After the integrations take place, we obtain:

\begin{multline}
    x_1 = -\frac{\text{sech}^2\left(\alpha  c_3 \left(t-c_4\right)\right)}{8 \alpha  c_2^2}
   \Bigg[c_3 c_2 \left(t-c_4\right) \sinh \left(2 \alpha  c_3
   \left(t-c_4\right)\right)\\
   +2 c_3 \left(\alpha  c_1
   \left(c_{10} \left(2 \alpha  c_3 \left(t-c_4\right)+\sinh
   \left(2 \alpha  c_3 \left(t-c_4\right)\right)\right)+2
   c_9\right)-c_7 \sinh \left(2 \alpha  c_3
   \left(t-c_4\right)\right)\right)\\
   +4 \alpha  c_6 c_2^2 t \cosh
   ^2\left(\alpha  c_3 \left(t-c_4\right)\right)\Bigg] + c_{11}.\label{eq:K_4_x1}
\end{multline}

and

\begin{multline}
    y_1 = \frac{1}{8\alpha  c_2^2} \Bigg[4 \alpha  c_2^2 c_5 \left(t-c_4\right)+2 c_3 \left(c_1
   \left(t-c_4\right)+2 c_8\right) \tanh \left(\alpha  c_3
   \left(t-c_4\right)\right)\\
   -2 \alpha  c_2 c_3 \left(c_{10}
   \left(2 \alpha  c_3 \left(t-c_4\right)+\sinh \left(2 \alpha 
   c_3 \left(t-c_4\right)\right)\right)+2 c_9\right)
   \text{sech}^2\left(\alpha  c_3 \left(t-c_4\right)\right)\Bigg] + c_{12}\label{eq:K_4_y1}
\end{multline}

where $c_{11}, c_{12}$ are new constants of integration. To summarize this section, the solution for the Killing magnetic curve, up to first order in $B_4$ is as follows:

\begin{align}
    x &= x_0 + B_4 x_1 + \mathcal{O}(B_4^2),\\
    y &= y_0 + B_4 y_1 + \mathcal{O}(B_4^2),\\
    z &= z_0 + B_4 z_1 + \mathcal{O}(B_4^2).
\end{align}

where $x_0, y_0, z_0$ are solutions for the geodesic equation and $x_1,y_1,z_1$ are given in Equations~(\ref{eq:K_4_x1},\ref{eq:K_4_y1},\ref{eq:K_4_z1}) respectively.

\subsection{Magnetic Trajectory by the Fifth Killing Vector Field}
Using the expression found for $F_{(5)}$ in Equation~(\ref{eq:kmc_F_5}) and Equation~(\ref{eq:kmc_def}), then turning them into the coordinate basis we obtain the following set of equations:

\begin{align}
 \ddot x - 2\alpha \dot x  \dot z &= B_5 x (-\dot y + \alpha y \dot z) \\
  \ddot y - 2\alpha \dot y  \dot z &= \frac {B_5}{2 \alpha} (2x \alpha \dot x + e^{2 \alpha z} \dot z + \alpha^2 (y^2-x^2) \dot z), \\
  \ddot z + \alpha (\dot x^2 + \dot y^2) e^{-2 \alpha z} &= -\frac {B_5}{2 \alpha} \left( \alpha^2 e^{-2 \alpha z}((y^2-x^2) \dot y + 2xy \dot x) + \dot y \right)
\end{align}

We will solve these equations of motion using a perturbation analysis. We write down $x = x_0 + B_5 x_1, y = y_0 + B_5 y_1, z = z_0 + B_5 z_1$ upto first order in $B_5$ and note that the function with zero indices are solutions of the geodesic equation of motion that we have found in Section~\ref{sec:geodesic_eom}. The differential equation for $x_1$ is found as follows:

\begin{align}
    \ddot x_1 - 2\alpha \dot z_0 \dot x_1 &= 2\alpha \dot x_0 \dot z_1 - x_0 \dot y_0 + \alpha x_0 y_0 \dot z_0,
\end{align}

we multiply this equation by $e^{-2\alpha z_0}$ and then integrate to find,

\begin{align}
    \dot x_1 &= 2\alpha c_1 z_1 e^{2\alpha z_0} - c_2 e^{2\alpha z_0} \int x_0 \mathrm dt - \frac 12 e^{2\alpha z_0} \int x_0 y_0 \partial_t e^{-2\alpha z_0} \mathrm dt + c_7 e^{2\alpha z_0}\label{eq:K_5_x1dot_plain}
\end{align}

where $c_7$ is a constant of integration. Now, let us turn our focus to $y_1$ whose differential equation reads as:

\begin{align}
    \ddot y_1 - 2\alpha \dot z_0 \dot y_1 &= 2\alpha \dot y_0 \dot z_1 + x_0 \dot x_0 + \frac{1}{2\alpha} \dot z_0 e^{2\alpha z_0} + \frac{\alpha}{2} (y_0^2 - x_0^2) \dot z_0,
\end{align}

we multiply with $e^{-2\alpha z_0}$ and integrate to find,

\begin{multline}
    \dot y_1 = 2\alpha c_2 z_1 e^{2\alpha z_0} + c_1 e^{2\alpha z_0}\int x_0 \mathrm dt \\
    + \frac{1}{2\alpha} z_0 e^{2\alpha z_0} - \frac{e^{2\alpha z_0}}{4}\int (y_0^2 - x_0^2)\partial_t e^{-2\alpha z_0} \mathrm dt + c_8 e^{2\alpha z_0},\label{eq:K_5_y1dot_plain}
\end{multline}

where $c_8$ is a constant of integration. Lastly we write down the differential equation for $z_1$:

\begin{align}
    \ddot z_1 - 2\alpha^2 e^{2\alpha z_0} (c_1^2 + c_2^2) z_1 + 2\alpha(c_1 \dot x_1 + c_2 \dot y_1) &= -\frac{\dot y_0}{2\alpha} - \frac{\alpha}{2}\Big[(y_0^2 - x_0^2)c_2 + 2 c_1 x_0 y_0\Big].
\end{align}

When we use the forms of $\dot x_1, \dot y_1$ in Equations~(\ref{eq:K_5_x1dot_plain},\ref{eq:K_5_y1dot_plain}) we obtain:

\begin{multline}
    e^{-2\alpha z_0} \ddot z_1 + 2\alpha^2 (c_1^2 + c_2^2) z_1 = -\frac{c_2}{2\alpha} - \frac{\alpha}{2}e^{-2\alpha z_0} [(y_0^2 - x_0^2)c_2 + 2 c_1 x_0 y_0]\\
    +\alpha c_1 \int x_0 y_0 \partial_t e^{-2\alpha z_0} \mathrm dt - 2\alpha (c_1 c_7 + c_2 c_8)\\
    - c_2 z_0 + \frac{\alpha}{2} c_2 \int (y_0^2 - x_0^2) \partial_t e^{-2\alpha z_0}.\label{eq:K_5_z1_ddot}
\end{multline}

We already calculated the homogeneous solution (see Equation~(\ref{eq:K_4_z1h})):

\begin{align}
    z_{1\text{h}} = -c_{10} + (c_9 + \alpha c_3 c_{10} (t - c_4)) \tanh(\alpha c_3 (t - c_4)),
\end{align}

where $c_9, c_{10}$ are constants of integration. We can calculate a part of particular solution, it is given as follows:

\begin{align}
    z_{1\text{p}} &= -\frac{c_2}{4\alpha^3 (c_1^2 + c_2^2)} - \frac 1\alpha \frac{c_1 c_7 + c_2 c_8}{c_1^2 + c_2^2} + f(t),
\end{align}

where $f(t)$ satisfies:

\begin{multline}
    e^{-2\alpha z_0} \ddot f + 2\alpha^2 (c_1^2 + c_2^2) f = -\frac{\alpha}{2}e^{-2\alpha z_0} [(y_0^2 - x_0^2)c_2 + 2 c_1 x_0 y_0]\\
    - c_2 z_0 +\alpha c_1 \int x_0 y_0 \partial_t e^{-2\alpha z_0} \mathrm dt
     + \frac{\alpha}{2} c_2 \int (y_0^2 - x_0^2) \partial_t e^{-2\alpha z_0} \mathrm dt.
\end{multline}

We simplify the equation and find:

\begin{align}
    e^{-2\alpha z_0} \ddot f + 2\alpha^2 (c_1^2 + c_2^2) f = -c_2 z_0 + \frac{\alpha}{2} c_1 c_2 \int x_0 \mathrm dt - \alpha \left(c_1^2 + \frac{c_2^2}{2}\right)\int y_0 \mathrm dt.
\end{align}

When we do the integrals and put the values of functions $x_0, y_0, z_0$ of the geodesic equation, we find the solutions as:

\begin{align}
    f &= \frac{A_5}{4}[2 \log \cosh(\tau) - \tau \tanh(\tau)] \nonumber \\
    &\quad + \frac{1}{2\alpha^2 (c_1^2 + c_2^2)} \left[\left(\frac{c_1 c_2 c_5}{2c_3} + \frac{c_6(2c_1^2 + c_2^2)}{2}\right)\tau - \frac{c_2}{2\alpha}\log\left(\frac{c_3^2}{c_1^2 + c_2^2}\right)\right]
\end{align}

where

\begin{align}
    A_5 &= \frac{1}{\alpha^3 (c_1^2 + c_2^2)} \left(c_2  + \frac{c_1^2 c_2}{2 (c_1^2 + c_2^2)} + \frac{2c_2 c_3 (2 c_1^2 + c_2^2)}{2(c_1^2 + c_2^2)}\right),
\end{align}

and remember that $\tau = \alpha c_3 (t-c_4)$. In the end, we find $z_1$ as $z_1 = z_{1\text{h}} + z_{1\text{p}}$ where

\begin{align}
    z_{1\text{h}} &= -c_{10} + (c_9 +  c_{10} \tau) \tanh(\tau),\\
    z_{1\text{p}} &= -\frac{c_2}{4\alpha^3 (c_1^2 + c_2^2)} - \frac 1\alpha \frac{c_1 c_7 + c_2 c_8}{c_1^2 + c_2^2} \nonumber \\
    &\quad +\frac{A_5}{4}[2 \log \cosh(\tau) - \tau \tanh(\tau)] \nonumber \\
    &\quad +\frac{1}{2\alpha^2 (c_1^2 + c_2^2)} \left[\left(\frac{c_1 c_2 c_5}{2c_3} + \frac{c_6(2c_1^2 + c_2^2)}{2}\right)\tau - \frac{c_2}{2\alpha}\log\left(\frac{c_3^2}{c_1^2 + c_2^2}\right)\right].
\end{align}

We would like to write $z_1$ as:

\begin{multline}
    z_1 = -c_{10} + (c_9 +  c_{10} \tau) \tanh(\tau) + A_1 + A_2\tau  
    + \frac{A_5}{4}[2 \log \cosh(\tau) - \tau \tanh(\tau)]\label{eq:K_5_z1},
\end{multline}

where

\begin{align}
    A_1 &=  -\frac{c_2}{4\alpha^3 (c_1^2 + c_2^2)} - \frac 1\alpha \frac{c_1 c_7 + c_2 c_8}{c_1^2 + c_2^2} - \frac{1}{2\alpha^2 (c_1^2 + c_2^2)} \frac{c_2}{2\alpha}\log\left(\frac{c_3^2}{c_1^2 + c_2^2}\right),\\
    A_2 &= \frac{1}{2\alpha^2 (c_1^2 + c_2^2)} \left(\frac{c_1 c_2 c_5}{2c_3} + \frac{c_6(2c_1^2 + c_2^2)}{2}\right).
\end{align}

All we need to now to finish this section, is to put the value of $z_1$ we have just provided into Equations~(\ref{eq:K_5_x1dot_plain},\ref{eq:K_5_y1dot_plain}) (for $\dot x_1$ and $\dot y_1$) and integrate them. After a few algebraic manipulations we obtain $\dot x_1, \dot y_1$ from Equations~(\ref{eq:K_5_x1dot_plain},\ref{eq:K_5_y1dot_plain}) as follows:

\begin{align}
    \dot x_1 &= 2\alpha c_1 z_1 e^{2\alpha z_0} - \frac 12 x_0 y_0 + \frac{1}{2}e^{2\alpha z_0} (c_1 c_6 - c_2 c_5)(t-c_4) + c_7 e^{2\alpha z_0},
\end{align}

and

\begin{multline}
    \dot y_1 = 2\alpha c_2 z_1 e^{2\alpha z_0} + \frac 12 \alpha z_0 e^{2\alpha z_0} \\
    + \frac{3c_1}{4}e^{2\alpha z_0}\int x_0 \mathrm dt + \frac{c_2}{4} e^{2\alpha z_0} \int y_0 \mathrm dt - \frac{y_0^2 - x_0^2}{4}.
\end{multline}

where we absorbed a constant in $c_7$. When we perform the integrations via Mathematica~\cite{mathematica} we obtain:

\begin{multline}
    x_1 = \frac{1}{4\left(c_1^2+c_2^2\right){}^2} \Bigg[-4 \alpha  A_5 c_1 \left(c_1^2+c_2^2\right) c_3^2 \left(t-c_4\right)\\
    +c_1
   \left(c_1^2+c_2^2\right) c_3 \left(8 \alpha  A_2 c_3 \left(t-c_4\right)+8 A_1+3
   A_5-4 c_{10}\right) \tanh \left(\alpha  c_3 \left(t-c_4\right)\right)\\
   +c_1
   \left(c_1^2+c_2^2\right) c_3 \left(\alpha  c_3 \left(A_5-4 c_{10}\right)
   \left(t-c_4\right)-4 c_9\right) \text{sech}^2\left(\alpha  c_3
   \left(t-c_4\right)\right)\\
   -8 A_2 c_1 \left(c_1^2+c_2^2\right) c_3 \log \left(\cosh
   \left(\alpha  c_3 \left(t-c_4\right)\right)\right)\\
   +4 A_5 c_1
   \left(c_1^2+c_2^2\right) c_3 \tanh \left(\alpha  c_3 \left(t-c_4\right)\right) \log
   \left(\cosh \left(\alpha  c_3 \left(t-c_4\right)\right)\right)\\
   +\frac{2 c_1 c_2 c_3
   \tanh \left(\alpha  c_3 \left(t-c_4\right)\right)}{\alpha ^3}\\
   -\frac{\left(\alpha 
   \left(c_1^2+c_2^2\right) c_5-c_1 c_3\right) \left(\alpha  \left(c_1^2+c_2^2\right)
   c_6-c_2 c_3\right) \log \left(\tanh \left(\alpha  c_3
   \left(t-c_4\right)\right)+1\right)}{\alpha ^3 c_3}\\
   +\frac{\left(\alpha 
   \left(c_1^2+c_2^2\right) c_5+c_1 c_3\right) \left(\alpha  \left(c_1^2+c_2^2\right)
   c_6+c_2 c_3\right) \log \left(1-\tanh \left(\alpha  c_3
   \left(t-c_4\right)\right)\right)}{\alpha ^3 c_3}-\\
   \frac{2 \left(c_1^2+c_2^2\right)
   \left(c_1 c_6-c_2 c_5\right) \log \left(\cosh \left(\alpha  c_3
   \left(t-c_4\right)\right)\right)}{\alpha ^2}\\
   +\frac{2 \left(c_1^2+c_2^2\right) c_3
   \left(c_1 c_6-c_2 c_5\right) \left(t-c_4\right) \tanh \left(\alpha  c_3
   \left(t-c_4\right)\right)}{\alpha }\\
   +\frac{4 \left(c_1^2+c_2^2\right) c_3 c_7 \tanh
   \left(\alpha  c_3 \left(t-c_4\right)\right)}{\alpha }\Bigg] + c_{11},\label{eq:K_5_x1}
\end{multline}

and

\begin{multline}
    y_1 = \frac{1}{8 \alpha ^3 \left(c_1^2+c_2^2\right){}^2 c_3} \Bigg[-8 \alpha ^4 A_5 c_2 \left(c_1^2+c_2^2\right) c_3^3 \left(t-c_4\right)\\
    +2 \alpha^3 c_2 \left(c_1^2+c_2^2\right) c_3^2 \left(\alpha  c_3 \left(A_5-4 c_{10}\right)
   \left(t-c_4\right)-4 c_9\right) \text{sech}^2\left(\alpha  c_3
   \left(t-c_4\right)\right)\\
   -16 \alpha ^3 A_2 c_2 \left(c_1^2+c_2^2\right) c_3^2 \log
   \left(\cosh \left(\alpha  c_3 \left(t-c_4\right)\right)\right)\\
   +2 c_3^2 \tanh
   \left(\alpha  c_3 \left(t-c_4\right)\right) \Big(\alpha ^3 c_2
   \left(c_1^2+c_2^2\right) \left(8 \alpha  A_2 c_3 \left(t-c_4\right)+8 A_1+3 A_5-4
   c_{10}\right)\\
   +4 \alpha ^3 A_5 c_2 \left(c_1^2+c_2^2\right) \log \left(\cosh
   \left(\alpha  c_3 \left(t-c_4\right)\right)\right)\\
   -2 \alpha ^2
   \left(c_1^2+c_2^2\right)+3 c_1 \left(\alpha ^2 \left(c_1^2+c_2^2\right) c_5
   t+c_1\right)+c_2 \left(\alpha ^2 \left(c_1^2+c_2^2\right) c_6 t+c_2\right)\\
   +\alpha
   ^2 \left(c_1^2+c_2^2\right) \log \left(\frac{c_3^2 \text{sech}^2\left(\alpha  c_3
   \left(t-c_4\right)\right)}{c_1^2+c_2^2}\right)+3 c_1^2 \log \left(\cosh
   \left(\alpha  c_3 \left(t-c_4\right)\right)\right)\\
   +c_2^2 \log \left(\cosh
   \left(\alpha  c_3 \left(t-c_4\right)\right)\right)-c_1^2+c_2^2\Big)\\
   +4 \alpha ^3
   \left(c_1^2+c_2^2\right) c_3^3 \left(t-c_4\right)+6 \alpha  c_1^2 c_3^3
   \left(c_4-t\right)+2 \alpha  c_2^2 c_3^3 \left(c_4-t\right)\\
   +\left(c_1 c_3-\alpha 
   \left(c_1^2+c_2^2\right) c_5\right){}^2 \log \left(\tanh \left(\alpha  c_3
   \left(t-c_4\right)\right)+1\right)\\
   -\left(\alpha  \left(c_1^2+c_2^2\right) c_5+c_1
   c_3\right){}^2 \log \left(1-\tanh \left(\alpha  c_3
   \left(t-c_4\right)\right)\right)\\
   -\left(c_2 c_3-\alpha  \left(c_1^2+c_2^2\right)
   c_6\right){}^2 \log \left(\tanh \left(\alpha  c_3
   \left(t-c_4\right)\right)+1\right)\\
   +\left(\alpha  \left(c_1^2+c_2^2\right) c_6+c_2
   c_3\right){}^2 \log \left(1-\tanh \left(\alpha  c_3
   \left(t-c_4\right)\right)\right)\\
   -6 \alpha  c_1 \left(c_1^2+c_2^2\right) c_3 c_5
   \log \left(\cosh \left(\alpha  c_3 \left(t-c_4\right)\right)\right)\\
   -2 \alpha  c_2
   \left(c_1^2+c_2^2\right) c_3 c_6 \log \left(\cosh \left(\alpha  c_3
   \left(t-c_4\right)\right)\right)\Bigg] + c_{12}\label{eq:K_5_y1}.
\end{multline}

To summarize, the solution for the Killing magnetic curve, up to first order in $B_5$ is as follows:

\begin{align}
    x &= x_0 + B_5 x_1 + \mathcal{O}(B_5^2),\\
    y &= y_0 + B_5 y_1 + \mathcal{O}(B_5^2),\\
    z &= z_0 + B_5 z_1 + \mathcal{O}(B_5^2).
\end{align}

where $x_0, y_0, z_0$ are solutions for the geodesic equation and $x_1,y_1,z_1$ are given in Equations~(\ref{eq:K_5_x1},\ref{eq:K_5_y1},\ref{eq:K_5_z1}) respectively. Last but not least, wihle making use of $x_1,y_1,z_1$, remember that $\tau = \alpha c_3 (t-c_4)$.

where $c_{11}, c_{12}$ are constants of integration.

\subsection{Magnetic Trajectory by the Sixth Killing Vector Field}

Using the expression found for $F_{(6)}$ in Equation~(\ref{eq:kmc_F_6}) and Equation~(\ref{eq:kmc_def}), then turning them into the coordinate basis we obtain the following set of equations:

\begin{align}
  \ddot x - 2\alpha \dot x \dot z &= - \frac{B_6}{2 \alpha}( 2y \alpha \dot y + e^{2 \alpha z} \dot z + \alpha^2 (x^2-y^2) \dot z) , \\
  \ddot y - 2\alpha \dot y \dot z &= B_6 y ( \dot x - \alpha x \dot z )  , \\
  \ddot z + \alpha(\dot x^2 + \dot y^2)e^{-2\alpha z} &= \frac{B_6}{2 \alpha} \left( \alpha^2 e^{-2 \alpha z}( (x^2-y^2)\dot x+2xy \dot y) + \dot x  \right) 
\end{align}

We will solve these equations of motion using a perturbation analysis. We write down $x = x_0 + B_6 x_1, y = y_0 + B_6 y_1, z = z_0 + B_6 z_1$ upto first order in $B_6$ and note that the function with zero indices are solutions of the geodesic equation of motion that we have found in Section~\ref{sec:geodesic_eom}. Because of the symmetry between $K_5$ and $K_6$ we do not need to solve the equations for this case explicitly. Just map $x\leftrightarrow y$ and $B_5 \leftrightarrow -B_6$. So the perturbative solution found in the previous section is valid for this case when the symmetry transformations are done.

\section{Conclusion}\label{sec:conlusion}

In this study, we investigated the $\mathbb H^3$ manifold which is a $(-\alpha)$-Kenmotsu manifold admitting six Killing vector fields. We have solved the geodesic equation of motion analytically. We calculated the motion of a charged particle under the magnetic field $B_i K_{(i)}$ upto first order in $B_i$ analytically using perturbation theory for all Killing vector fields $K_{(i)}$. We have put a scaling factor $B_i$ in front of Killing vectors to manage strength of magnetic field and make it amenable to perturbative analysis. We used units where $q/m = -1$, however notice that the unit of each $B_i$ is not necessarily the unit of magnetic field. We also proved that 3-dimensional $(\alpha)$-Kenmotsu manifolds cannot have any magnetic vector field in the direction of their Reeb vector fields. This result is interesting because most of the studies related with magnetic curves of almost contact manifolds in literature deals with Reeb magnetic curves.

In Appendix~\ref{app:poincare} we have plotted the geodesic solution, 1st order perturbative result, and numerical solution for the second Killing magnetic field (see Section~\ref{sec:K_2}) where the $y$-coordinate is suppressed, e.g. $y=0$. In Appendix~\ref{app:KillingH3} we have shown calculation steps explicitly for determining Killing vector fields.

\section*{Acknowledgements}

We would like to thank the anonymous referee for providing constructive criticisms and clarifications.



\begin{appendices}

\section{Representation of a Solution in the Poincaré Disk}\label{app:poincare}

In this section, we make connection with the hyperbolic manifold we used ($\mathbb H^3$) and the Poincaré Disk. The metric we have is:

\begin{equation}
    \mathrm ds^2 = e^{-2\alpha z} (\mathrm dx^2 + \mathrm dy^2) + \mathrm dz^2.
\end{equation}

By multiplying the metric with $\alpha^2$ we obtain:

\begin{equation}
    \mathrm d\sigma^2 = e^{-2z} (\mathrm dx^2 + \mathrm dy^2) + \mathrm dz^2,
\end{equation}

\noindent where we map $\alpha x \to x, \alpha y\to y, \alpha z \to z$ (This is equivalent to setting $\alpha = 1$). If we suppress the $y$ coordinate, we obtain:

\begin{equation}
    \mathrm d\Sigma^2 = e^{-2z} \mathrm dx^2 + \mathrm dz^2.
\end{equation}

Let us define $y = e^z$ (this is \emph{not} the original `$y$' coordinate.) Then the metric becomes:

\begin{equation}
    \mathrm d\Sigma^2 = \frac{\mathrm dx^2 + \mathrm dy^2}{y^2},
\end{equation}

\noindent which is the metric on the Poincaré half plane \cite{PoincareUpperHalfPlane}. Considering this plane defined by $\{(x,y) \mid x \in \mathbb R, y \in \mathbb R^+\}$ as the upper half complex plane with $z = x + iy$, using the following Möbius transformation (where this particular form is known as Cayley transformation~\cite{cayley_transform}), we map upper complex plane into a unit disk:

\begin{equation}
    \omega(z) = \frac{z-i}{z+i}.
\end{equation}

The Cartesian coordinates ($\omega_x$ is the real part, $\omega_y$ is the imaginary part) that correspond to $\omega(x+iy)$ is then calculated as:

\begin{align}
    \omega_x = \frac{x^2 + y^2 - 1}{x^2 + (y+1)^2}, \quad \omega_y = -\frac{2x}{x^2 + (y+1)^2}.
\end{align}

Representation of solutions in the Poincaré Disk is useful, because it allows us to visualize the asymptotic behavior of solutions as $t\to\pm\infty$. In a way, it may be seen similar to Penrose diagrams~\cite{PenroseDiagram} in general relativity. See Figure~\ref{fig:k2plot} for a comparison of our 1st order result with the numerical solution.

\begin{figure}[H]
    \centering
    \includegraphics{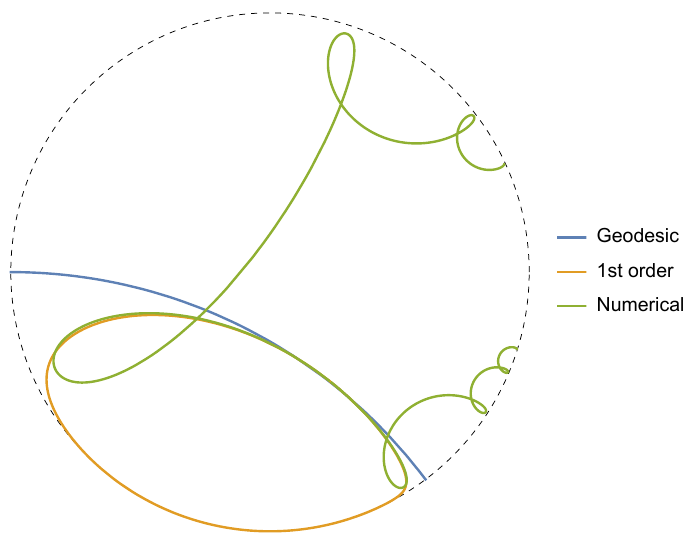}
    \caption{Here we see three plots. 1) The geodesic solution, 2) The 1st order solution, 3) The numerical solution. We see that the 1st order solution properly approximates the numerical solution when $|t|$ is less than some upper-bound, but later the numerical solution takes over and displays complex behavior. In the plot, $t\in(-30,30)$. The parameters are as follows: $c_1 = 1, c_2 = 0, c_3 = 1, c_4 = 0, c_5 = 1, c_6 = 0, c_7 = 1, B_2 = 0.1$. and for the compatible numerical solution $x(0)=1.1, x'(0) = 0.95, z(0)=0, z'(0) = -0.1$.}
    \label{fig:k2plot}
\end{figure}

\section{Killing Vector Fields of $\mathbb H^3$ }\label{app:KillingH3}
We show explicit calculation steps for Killing vector fields by using a general ansatz such as $K=K^{x}(x,y,z) \partial _{x}+K^{y}(x,y,z)\partial _{y}+K^{z}(x,y,z)\partial _{z}$. Following partial differential equations were obtained from the Killing
equation \footnote{ Note that general vector field can be written in index form as $K=K^\mu \partial_{\mu} $ and we lower the indices by using the metric so that $K_{\alpha}=g_{\mu \alpha} K^\mu$ can be used in \eqref{Killing_general}. $K^\mu$ coefficients correspond to the functions in our ansatz, $K^{x}(x,y,z)$ etc.} \eqref{Killing_general}

\begin{equation}\label{BKilling}
\begin{array}{c}
 -\alpha K^{z}( x,y,z) + \partial_x K^{x} ( x,y,z) =0 \\ 
 -\alpha K^{z}( x,y,z) + \partial_y K^{y} ( x,y,z) =0  \\ 
e^{-2\alpha z }\partial_z K^{y} (x,y,z)+ \partial_y K^{z} ( x,y,z) =0 \\ 
 \partial_y K^{x} (x,y,z)+ \partial_x K^{y}(x,y,z) =0 \\ 
e^{-2\alpha z } \partial_z K^{x} (x,y,z) + \partial_x K^{z} (x,y,z) =0 \\ 
 \partial_z K^{z}(x,y,z) =0
\end{array}
\end{equation}

 From the last equation we have $K^{z}\left( x,y,z\right) =c+K^{z}\left( x,y\right) $.
Inserting this in equations \eqref{BKilling} and writing with a more compact notation (i.e. assuming $K^i$ s depend on $x,y,z$) we get,

\begin{equation}
-\alpha \left( c+K^{z}\right) +\partial_x K^{x}=0  \label{B1}
\end{equation}

\begin{equation}
-\alpha \left( c+K^{z}\right) +\partial_y K^{y}=0  \label{B2}
\end{equation}

\begin{equation}
\partial_y K^{z}+ e^{-2\alpha z}\partial_z K^{y}=0  \label{B3}
\end{equation}

\begin{equation}
\partial_y K^{x}+\partial_x K^{y}=0  \label{B4}
\end{equation}

\begin{equation}
\partial_x K^{z}+e^{-2\alpha z}\partial_z K^{x}=0  \label{B5}
\end{equation}

 From \eqref{B1} and \eqref{B2}

\begin{equation}
\partial_x K^{x}=\partial_y K^{y}  \label{B6}
\end{equation}

 Taking $\partial_z$ of \eqref{B1} and \eqref{B2}

\begin{eqnarray}
\partial_z\partial_x K^{x}=0 \ , \quad \partial_z\partial_y K^{y}=0   \label{B8}
\end{eqnarray}

 Taking $\partial_x$ of \eqref{B3} and $\partial_y$ of \eqref{B5}

\begin{equation} \label{B9}
    \partial_x\partial_y K^{z}=-e^{-2\alpha z}\partial_z \partial_x K^{y}=-e^{-2\alpha z}\partial_z \partial_y K^{x}
\end{equation}

 From \eqref{B9} and \eqref{B4}

\begin{align} \label{B10}
    &\partial_z\partial_x K^{y}=\partial_z\partial_y K^{x} \ , \quad \partial_y K^{x}= -\partial_x K^{y} \implies \partial_z\partial_y K^{x}= -\partial_z \partial_x K^{y} \nonumber \\
    &\implies \partial_z\partial_x K^{y}=-\partial_z \partial_x K^{y} \implies \partial_z\partial_x K^{y}=0 \ , \quad \partial_z \partial_y K^{x}=0
\end{align}

 Equations \eqref{B8}  and \eqref{B10} dictates separation of $z$ variable for $K^{x}\left( x,y,z\right)$ and $K^{y}\left(x,y,z\right)$. Hence, the general solution of $K^{x}$ and $K^{y}$ should be in the following form

\begin{equation} \label{B11}
K^{x}\left( x,y,z\right) =h_{1}(z)+f_{1}(x,y) \quad , \quad K^{y}\left( x,y,z\right)=h_{2}(z)+f_{2}(x,y)
\end{equation}

 Inserting these forms back into Killing equations \eqref{BKilling} we get

\begin{equation}
-\alpha \left( c+K^{z}\right) +\partial_x f_{1}=0  \label{B12}
\end{equation}

\begin{equation}
-\alpha \left( c+K^{z}\right) +\partial_y f_{2}=0  \label{B13}
\end{equation}

\begin{equation}
e^{-2\alpha z}h_{2}'(z)+\partial_y K^{z}=0  \label{B14}
\end{equation}

\begin{equation}
\partial_y f_{1}+\partial_x f_{2}=0  \label{B15}
\end{equation}

\begin{equation}
e^{-2\alpha z}h_{1}'(z)+\partial_x K^{z}=0  \label{B16}
\end{equation}

 From \eqref{B12} and \eqref{B13} we have

\begin{equation}
\partial_x f_{1}=\partial_y f_{2}  \label{B17}
\end{equation}

 Using \eqref{B17} and \eqref{B15} we get the following,

\begin{eqnarray}\label{B18}
& \quad \partial_{x}^{2}f_{1}=\partial_x \partial_y f_{2} \quad , \quad \partial_{y}^{2}f_{1}=-\partial_x \partial_y f_{2} \nonumber \\ 
& \quad \partial_{y}^{2}f_{2}=\partial_x \partial_y f_{1}  \quad , \quad  \partial_{x}^{2}f_{2}=-\partial_x \partial_y f_{1}  \nonumber \\ 
&\implies \partial_{x}^{2} f_{1}+\partial_{y}^{2} f_{1}=0 \quad , \quad  \partial_{x}^{2} f_{2}+\partial_{y}^{2} f_{2}=0 
\end{eqnarray}

 The last two equations of \eqref{B18} are Laplace equations in 2-dimensions. We can write a general solution for those equations as \footnote{A general solution to the Laplace equation is given as $f(x,y) =( A\cosh ( \lambda x) +B\sinh ( \lambda x) )( C\cos( \lambda y) +D\sin ( \lambda y) )$. But this solution leads to inconsistencies in Killing equations. Therefore, we use a more simplified general solution which can potentially solve the Killing equations.}

\begin{equation}\label{B21}
f_{1}\left( x,y\right) =c_{1}(x^{2}-y^{2})+c_{2}xy+c_{3}x+c_{4}y+c_{5}
\end{equation}

\begin{equation}\label{B22}
f_{2}\left( x,y\right) =c_{6}(x^{2}-y^{2})+c_{7}xy+c_{8}x+c_{9}y+c_{10}
\end{equation}

 Inserting \eqref{B21} and \eqref{B22} in equations \eqref{B12}, \eqref{B13} and \eqref{B15}

\begin{equation}
2xc_{1}+yc_{2}+c_{3}-\alpha c-\alpha K^{z}=0  \label{B23}
\end{equation}

\begin{equation}
-2yc_{6}+xc_{7}+c_{9}-\alpha c-\alpha K^{z}=0  \label{B24}
\end{equation}

\begin{equation}
\left( c_{2}+2c_{6}\right) x+\left( c_{7}-2c_{1}\right) y+c_{4}+c_{8}=0
\label{B25}
\end{equation}

 From these equations we get 

\begin{equation}
c_{3}=c_{9}=\alpha c \ , \quad  c_{2}=-2c_{6} \ , \quad c_{7}=2c_{1} \ , \quad c_{4}=-c_{8}
\label{B26}
\end{equation}

 Equations \eqref{B23} and \eqref{B24} gives information about the general form of $
K^{z}\left( x,y\right) $ which should be linear with respect to $x$ and $y$. Thus,

\begin{equation}
K^{z}\left( x,y\right) =c_{12}x+c_{13}y  \label{B27}
\end{equation}

 Using \eqref{B27} for $K^{z}\left( x,y\right)$ in Killing equations we get,

\begin{equation}
2xc_{1}+yc_{2}+c_{3}-\alpha c-x\alpha c_{12}-y\alpha c_{13}=0  \label{B28}
\end{equation}

\begin{equation}
2yc_{6}-xc_{7}-c_{9}+\alpha c+x\alpha c_{12}+y\alpha c_{13}=0  \label{B29}
\end{equation}

\begin{equation}
c_{13}+e^{-2\alpha z}h_{2}'(z) =0  \label{B30}
\end{equation}

\begin{equation}
c_{4}+x\left( c_{2}+2c_{6}\right) +y\left( -2c_{1}+c_{7}\right) +c_{8}=0
\label{B31}
\end{equation}

\begin{equation}
c_{12}+e^{-2\alpha z}h_{1}'(z) =0  \label{B32}
\end{equation}

 Solving these all together we get

\begin{equation}
c_{12}=\frac{c_{7}}{\alpha }=\frac{2c_{1}}{\alpha } \ , \quad  c_{13}=\frac{c_{2}}{
\alpha }=\frac{-2c_{6}}{\alpha } \ , \quad  c_{3}=c_{9}=\alpha c \ , \quad  c_{8}=-c_{4} \label{B33}
\end{equation}

 From \eqref{B30} and \eqref{B32}

\begin{equation}
h_{1}\left( z\right) =-\frac{c_{12}}{2\alpha }e^{2\alpha z} \ , \quad   h_{2}\left( z\right) =-\frac{c_{13}}{2\alpha }e^{2\alpha z}   \label{B34}
\end{equation}

 After relabeling the constants and using the relations obtained in \eqref{B33}
and \eqref{B34} we reach the final general solution for $K^{x},K^{y},K^{z}$

\begin{equation}
\begin{array}{c}
K^{x}=\left( \frac{\alpha }{2}\left( x^{2}-y^{2}\right) -\frac{e^{2\alpha z}}{2\alpha }\right) c_{1}+\alpha xy c_{2}+\alpha xc_{3}+yc_{4}+c_{5} \\ 
K^{y}=\left( \frac{\alpha }{2}\left( y^{2}-x^{2}\right) -\frac{e^{2\alpha z}}{2\alpha }\right) c_{2}+\alpha xy c_{1}+\alpha yc_{3}-x c_{4}+c_{6} \\ 
K^{z}=c_{1}x+c_{2}y+c_{3}
\end{array}
\label{B36}
\end{equation}

 Recall that we have taken the general ansatz $K=K^{x}\partial _{x}+K^{y}\partial _{y}+K^{z}\partial _{z}$  for Killing vector fields. Using the functions given in \eqref{B36} and adjusting the constants accordingly ($c_{i}=1$  and all $\left.\ c_j\right|_{j\neq i}=0$) we obtain 6 independent Killing vectors.

\begin{align}
     &\mathbf{K_1} = \partial_x, \quad \mathbf{K_2} = \partial_y, \quad \mathbf{K_3} = y\partial_x - x\partial_y, \quad  \mathbf{K_4} = \alpha x \partial_x + \alpha y \partial_y + \partial_z, \nonumber \\
       &\mathbf{K_5} = \left( \frac{\alpha }{2}\left( x^{2}-y^{2}\right) -\frac{e^{2 \alpha z}}{2\alpha }\right) \partial _{x}+\alpha x y \partial _{y}+x\partial _{z} , \nonumber \\
    &\mathbf{K_6} = \alpha x y \partial _{x}+\left( \frac{\alpha }{2}\left(y^{2}-x^{2}\right) -\frac{e^{2 \alpha z}}{2\alpha }\right) \partial_{y}+y\partial _{z}    
\end{align}

 Note that maximally symmetric spaces have $n\left( n+1\right) /2$ independent Killing vectors, where $n$ is the dimension of the space. Therefore, having 6 independent Killing vectors $\mathbb H^3$ is a maximally symmetric space.
\end{appendices}



\end{document}